\documentclass[sigplan,screen]{acmart}
\usepackage[colorinlistoftodos,prependcaption,textsize=tiny]{todonotes}

\usepackage{xcolor}
\definecolor{pastelred}{rgb}{1.0, 0.41, 0.38}

\AtBeginDocument{%
  \providecommand\BibTeX{{%
    \normalfont B\kern-0.5em{\scshape i\kern-0.25em b}\kern-0.8em\TeX}}}

\usepackage[linesnumbered]{algorithm2e}
  
\usepackage{amsmath}
\usepackage{amssymb}
\usepackage{graphicx}
\usepackage{hyperref}
\usepackage{proof}
\usepackage{listings}
\usepackage{mycommands}
\usepackage{balance}
\setcopyright{acmcopyright}
\acmPrice{15.00}
\acmDOI{10.1145/3426425.3426944}
\acmYear{2020}
\copyrightyear{2020}
\acmSubmissionID{sle20main-p78-p}
\acmISBN{978-1-4503-8176-5/20/11}
\acmConference[SLE '20]{Proceedings of the 13th ACM SIGPLAN International Conference on Software Language Engineering}{November 16--17, 2020}{Virtual, USA}
\acmBooktitle{Proceedings of the 13th ACM SIGPLAN International Conference on Software Language Engineering (SLE '20), November 16--17, 2020, Virtual, USA}

\begin{document}

\title{A Semantic Framework for PEGs}

\author{Sérgio Queiroz de Medeiros}
\email{sergiomedeiros@ect.ufrn.br}
\orcid{0000-0002-0759-0926}
\author{Carlos Olarte}
\authornote{Carlos Olarte was funded by CNPq.}
\email{carlos.olarte@gmail.com}
\orcid{0000-0002-7264-7773}
\affiliation{%
  \institution{ECT, Federal University of Rio Grande do Norte}
  \city{Natal}
	\country{Brazil}
}

\renewcommand{\shortauthors}{Sérgio Queiroz de Medeiros and Carlos Olarte}

\begin{abstract}
Parsing Expression Grammars (PEGs) are a recognition-based formalism  which allows to describe the syntactical and the lexical elements of a language. The main difference between Context-Free Grammars (CFGs) and PEGs relies on the interpretation of the choice operator: while the CFGs' unordered choice $e \mid e'$ is interpreted as the union of the languages recognized by $e$ and $e'$, the PEGs' prioritized choice $e\slash e'$ discards $e'$ if $e$ succeeds.  Such subtle,  but important difference,   changes  the language 
recognized and yields more efficient parsing algorithms.  This paper proposes a rewriting logic  semantics for PEGs. We start
with a rewrite theory giving meaning to the usual constructs in PEGs. Later, we show that cuts, a mechanism for controlling  backtracks in PEGs, finds  also a natural representation in our framework. 
 We   generalize such mechanism, allowing for both  local and global cuts with a precise, unified and formal semantics. Hence, our work strives at better understanding and controlling backtracks in parsers for PEGs. The semantics we propose is executable and, besides being a parser with modest efficiency, it can be used as a playground to test different optimization ideas.  More importantly, it is a mathematical tool that can be used for different analyses.

\end{abstract}



\maketitle

\section{Introduction}\label{sec:intro}
Parsing Expression Grammars (PEGs) \cite{DBLP:conf/popl/Ford04} are the core of several widely used parsing tools.
Visually, the description of a PEG is similar to the description of a Context-Free Grammar (CFG).
Unlike CFGs, PEGs have a deterministic ordered choice operator, which allows a limited form
of backtracking. This makes PEGs a suitable formalism for representing deterministic Context-Free Languages,
i.e., the LR(k) languages. Another key difference between both formalisms is 
the presence of syntactic predicates, that allow PEGs to  describe the lexical elements
of a language. PEGs were conceived as a formalism to recognize strings, while
CFGs are commonly used to generate strings.  Hence, it may not be trivial to determine the language described by a PEG.

Although PEGs ordered choice operator avoids ambiguities when writing a grammar,
it also poses some difficulties. To correctly recognize a language, the user of a PEG-based
tool needs to  be careful about
  the ordering of the alternatives in a choice $e_1 \;/\; e_2$,
as $e_2$ will never match a string $x$ when $e_1$ matches a prefix of $x$.
Regarding performance, PEGs local backtracking may still impose a performance
drawback. Thus, when possible, it is desirable to avoid the backtracking
associated with a choice.

In this paper we study the problem of backtracking in PEGs through the lens of a formal approach based on the rewriting logic (RL)  \cite{meseguer-rltcs-1992}
framework. RL can be seen as a flexible and general model of computation where important properties of the modeled system can be specified and proved. 
We thus provide both:  a formal foundation for better understanding  and controlling backtracks in PEGs; 
and tools that can help the user  to check whether her grammar complies with the intended meaning or not. As an interesting side effect, we show that our  specification is not only a (correct by construction)
recognition-based algorithm, but also a derivative parser, able to generate strings from a grammar. \\

\noindent\textbf{Plan and contributions. }
After recalling   PEGs in \S \ref{sec:pegs}, we start
in \S \ref{sec:rw}  with a rewrite theory modeling the natural semantics rules  for PEGs originally proposed in \cite{DBLP:journals/scp/MascarenhasMI14}. Hence, we  obtain  a formal model of the  derivability relation in PEGs. 
Due to  the  $\epsilon$ representational distance typical of  RL specifications, the model and the actual system are very close, thus making it easier to reason  
about grammars in our framework. 
 
Inspired by a small-step semantics approach, 
\S \ref{sec:moreefficient} proposes an alternative (and equivalent) rewrite theory that eliminates 
  the (unnecessary)   non-deterministic steps of the first one. Both theories are 
executable in Maude \cite{DBLP:conf/maude/2007},
an efficient  rewriting engine. We  compare the two theories and   show that the second one is more amenable for automatic verification techniques. 
Hence, besides being a formal tool for reasoning  about PEGs, the proposed specification is actually a \emph{correct by construction} parser  with modest performance. 

Cut operators \cite{DBLP:conf/paste/MizushimaMY10} have been introduced in PEGs to reduce the number of backtracks and improve efficiency. However, the semantics of these operators has  remained  informal in the literature.  Using RL, 
\S \ref{sec:sym}  gives a precise meaning to  cuts.   The proposed semantics  makes evident why less memory is needed when cut annotations are added to a grammar. More importantly, such formal account allows us to generalize the concept of cuts 
from \emph{local} 
to \emph{global} cuts in \S \ref{sec:global}. 
In some cases, global cuts may save  more computations than the cuts proposed in 
\cite{DBLP:conf/paste/MizushimaMY10}. 
In \S \ref{sec:bound} and \S \ref{sec:simpl}
 we  show that local and global cuts can coexist coherently with a clear and uniform semantics. \S \ref{sec:bench} reports some benchmarks on grammars annotated with   cuts. 
 
 The machinery developed here can be leveraged to perform other analyses in PEGs. We explore one of such analyses in
  \S \ref{sec:symbolic}, where
  we report on a preliminary attempt to 
  use  our specification  to not only recognize a given string but, symbolically, produce all the possible strings (up to a bounded length)  from a grammar.
This kind of analyses can be useful to highlight unexpected behaviors when a PEG is designed, 
and also  to explore optimization ideas.

\S \ref{sec:related}  discusses  related work and 
\S \ref{sec:conc} concludes the paper.
The companion appendix contains detailed proofs  of the main results. 
All the rewrite theories proposed here and the benchmarks reported in \S \ref{sec:bench}
can be found (and reproduced) 
with the Maude and script files 
available at the  public repository  \url{https://github.com/carlosolarte/RESPEG}.

\section{Parsing Expression Grammars}\label{sec:pegs}

From now on, we shall write PEG to refer to the following definition of  parsing expression grammars. 

\begin{definition}[Syntax]
A   PEG $G$ is a tuple $(V,T,P, e_{\iota})$, where
$V$ is a finite set of non-terminals, $T$ is a finite set of terminals,
$P$ is a total function from non-terminals to parsing expressions 
and $e_{\iota}$ is the initial parsing expression. 
We shall use A,B,C to range over elements in $V$ and $a,b,c$ to range over elements in $T$. 
Parsing expressions, ranged over by $e, e_1, e_2, $ etc., are built from the syntax:
\[
\qquad e  :=  \epsilon ~\mid~ A ~\mid~ a ~\mid~ e_1 \ e_2 ~\mid~ e_1 \slash e_2 ~\mid~ e* ~\mid~ ! e 
\]
\end{definition}
We   assume that $V$ is partitioned into two (disjoint) sets $\Vlex$ and  $\Vsyn$, where $\Vlex$ is the set of non-terminals that match lexical elements,
also known as tokens, and $\Vsyn$ represents the non-terminals
that match syntactical elements.

The empty parsing expression is $\epsilon$. The expression  $e_1 \ e_2$ stands for the (sequential) composition of $e_1$ and $e_2$. 
The expression $e_1 \slash e_2$ denotes an ordered choice. 
The repetition of  $e$ is written as $e*$. 
The look-ahead or negative predicate is written as  $!e$. A predicate $!e$ tests if the expression $e$ matches the input, without consuming it. Predicates are handy to describe lexical elements and to act as   guards in  choice alternatives.   

The function $P$ 
maps non-terminals into parsing expressions.
$P(A)$ denotes the expression associated to $A$. Alternatively,  $P$ can be seen as a set of rules of the form 
$A \leftarrow e$. 

Strings are built from terminal symbols and 
$\epsilon$ denotes the empty string. 
We shall use $x,y,w$ to range over (possibly empty) strings and $xy$ denotes the concatenation of $x$ and $y$. \\

\noindent\textbf{Semantics.}
A parsing expression $e$, when applied to an
input string $x$, either succeeds or fails.
When the matching of $e$ succeeds, it consumes a prefix of the input. Such prefix can be the empty string $\epsilon$ (and nothing is consumed). 

We define the states of a PEG parser as follows. 
\begin{definition}[States]\label{def:states}
Parsing states are built from 

\[
S ::= \PEGstateD ~\mid~ \Lfail ~\mid~ x
\]
\end{definition}

 In $\PEGstateD$, the  expression $e$ in the context of the PEG $G$ is matched against the  string $x$. The 
expression $\Lfail$
represents a failed  attempt of matching\footnote{Failing states can also take the form $(\Lfail, y)$ where $y$ is the suffix of the input that could not be recognized. For the sake of presentation, we shall ignore here the suffix $y$.  }. 
 The state $x$   represents the successful matching of an expression returning the suffix $x$. 


\begin{figure}
\resizebox{.45\textwidth}{!}{$
\begin{array}{c}
\infer[\mylabel{empty}]{G[\varepsilon]  \; x \Lp x}{}
\qquad
\infer[\mylabel{var}]{G[A]  \; x \Lp S}{G[P(A)]  \; x \Lp S}\\\\
\infer[\mylabel{term.1}]{G[a]  \; ax \Lp x}{}
\qquad
\infer[\mylabel{term.2}]{G[b]  \; ax \Lp \Lfail}{b \ne a}
\\\\
\infer[\mylabel{term.3}]{G[a]  \; \varepsilon \Lp \Lfail}{} \\\\
\infer[\mylabel{seq.1}]{G[e_1 \; e_2]  \; x \Lp S}{G[e_1]  \; x \Lp y   \fivespaces  G[e_2]  \; y \Lp S}
\qquad			
\infer[\mylabel{seq.2}]{G[e_1 \; e_2]  \; x \Lp \Lfail}{G[e_1]  \; x \Lp \Lfail}\\\\
\infer[\mylabel{ord.1}]{\Matgk{e_1 \;\slash\; e_2}{x}{} \Lp y}{\Matgk{e_1}{x}{} \Lp y}
\qquad
\infer[\mylabel{ord.2}]{\Matgk{e_1 \;\slash\; e_2}{x}{} \Lp S}{\Matgk{e_1}{x}{} \Lp \Lfail \fivespaces \Matgk{e_2}{x}{} \Lp S}
\\\\
\infer[\mylabel{not.1}]{G[!e]  \; x \Lp x}{G[e] \; x \Lp \Lfail}
\qquad			
\infer[\mylabel{not.2}]{G[!e]  \; x \Lp \Lfail}{G[e] \;  x \Lp y}
\\\\
\infer[\mylabel{rep.1}]{G[e*]  \; x \Lp x}{G[e]  \; x \Lp \Lfail}
\qquad
\infer[\mylabel{rep.2}]{G[e*]  \; x \Lp z}{
 \deduce{G[e]  \; x \Lp {y}}{}
 &
 \deduce{G[e*]  \; y \Lp z}{}
}
\end{array}
$}

	\caption{Semantics of PEGs. $S$ is a parsing state denoting either $\Lfail$ or a string $x$ (Definition  \ref{def:states}).}
	\label{fig:sempeglab}
\end{figure}

\begin{definition}[Semantics]
 The reduction relation $\Lp$ is the least binary relation 
 on parsing states 
satisfying the rules in Fig. \ref{fig:sempeglab}. 
The language recognized by $e$ in the context of $G$ is the set 
 $
\Lang{e, G} = \{ x ~\mid~ \PEGstate{G}{e}{x} \Lp y \}
$. 
 Two parsing expressions are equivalent (in the context of $G$), notation $e \equiv e'$, if $\Lang{e,G} = \Lang{e',G}$. 
\end{definition}


Let us dwell upon the rules in Fig. \ref{fig:sempeglab}. 
If 
$\PEGstateD\Lp y$, the  expression $e$     consumes a (possibly empty) prefix of $x$
and returns the remaining suffix $y$. 
For instance, rule {\bf term.1} consumes the non-terminal $a$ in $ax$ and returns $x$. 
If 
$\PEGstateD \Lp \Lfail$,
$e$ fails to match the input $x$ (see  {\bf term.2} and {\bf term.3}).

In $\PEGstate{G}{e_1 \ e_2}{x} $, either $e_1$ fails to match $x$ and  the whole expression fails ({\bf seq.2})
or $e_1$ succeeds on $x$ and the final result depends on the outcome of  matching   $e_2$ against the remaining suffix $y$ ({\bf seq.1}).

When the ordered choice $e_1 \slash e_2$
is applied to $x$, 
if $e_1$ matches $x$,  then the alternative $e_2$ is discarded ({\bf ord.1} ). (Note the difference w.r.t.   CFGs). On the other side, if $e_1$ fails, 
a backtrack  is performed and $e_2$ is applied to the string $x$, regardless of whether  $e_1$ consumed part of it or not ({\bf ord.2}).

The look-ahead operator $! e$ (or negative predicate) 
fails when $e$ succeeds ({\bf not.2}) and succeeds (not consuming any input) when $e$ fails ({\bf not.1}). Note that the expression $!e$ does not consume any prefix of the input. 

The repetition $e*$ greedily matches $e$ against the input. If $e$ 
fails on $x$, $e*$ does not consume any input ({\bf rep.1}). Rule {\bf rep.2}
specifies the recursive case where $e*$ continues matching  the suffix $y$. Note that this operator is different from the usual one in regular expressions: $\PEGstate{G}{a*}{ xy} \Lp y$ iff $x$ is a (possibly empty)
string containing only the terminal symbol $a$ and $y$ starts with $b\neq a$. This operator can be in fact derived from the others: let $e$ be a parsing expression and $A_e$ be a distinguished non-terminal symbol. Then, $e*$ is equivalent to $A_e$ where $A_e\leftarrow e \slash \epsilon $ \cite{DBLP:conf/popl/Ford04}. 
We shall keep this operator in the syntax since it greatly simplifies some examples. 

Clearly,  if  $\PEGstate{G}{e}{x} \Lp y$ then  $y$ is a suffix of $x$. In order to guarantee termination on $\Lp$, a well-formedness condition on grammars is assumed \cite{DBLP:conf/popl/Ford04}:  there are no left-recursive rules  such as $A \leftarrow  A~ e_1 \slash e_2$. Moreover,  there are no expressions of the form  $e*$ where $e$ succeeds on some input $x$ not consuming any prefix of it (\cite{DBLP:journals/fuin/Redziejowski09}). 

Unlike CFGs,  PEGs are  \emph{deterministic} \cite{DBLP:conf/popl/Ford04}. This means that   $\Lp$ is a function.

Without loss of generality, our analyses consider only PEGs that satisfy the \emph{unique token prefix} condition~\cite{DBLP:journals/scp/MedeirosJM20}. Roughly,  tokens of the grammar are described by non-terminals $A \in \Vlex$. Hence, at most one  non-terminal $A \in \Vlex$ matches a prefix of the current input. This is the typical behavior of parsing tools that have a separate lexer (e.g., yacc gets  tokens from lex for a given input).
\begin{definition}[Unique token prefix]\label{def:unique}
Let $G = (V,T,P,e_{\iota})$, $V = \Vlex \cup \Vsyn$, 
$A,B \in \Vlex$ s.t. $A\neq B$,  $a \in T$ and $xy$ be a string.  $G$ has the unique token prefix property iff
$\PEGstate{G}{A}{axy} \Lp y$ implies that   $\PEGstate{G}{B}{axy} \Lp \Lfail$.

\end{definition}

\section{A Rewriting Logic Semantics for PEGs}\label{sec:rw}

This section proposes an executable rewriting logic semantics for PEGs that we later use to control backtracks (\S \ref{sec:sym}) and 
as a basis for a  derivative parser (\S \ref{sec:symbolic}). 

Rewriting Logic (RL) \cite{meseguer-rltcs-1992} is a general model of computation where proof systems, semantics of programming languages and, in general, transition systems can be specified and verified. 
 RL can be seen as a logic of change that can naturally deal with states and  concurrent computations. 
     The reader can find a  detailed survey of RL in  \cite{meseguer-twenty-2012} and   \cite{DBLP:journals/jlap/DuranEEMMRT20}.

 In the following, we briefly  introduce  the main concepts of RL needed to understand   this paper, and, at the same time, we  introduce step by step the proposed semantics for PEGs. 
  For the sake of readability,  we shall adopt in most of the cases the notation of Maude \cite{DBLP:conf/maude/2007}, a high-level language   that
supports membership equational logic and rewriting logic specifications. Thanks to its efficient rewriting
engine and its metalanguage capabilities, Maude turns out to be an excellent tool
for creating executable environments of various logics and  models of computation. This will make our specification \emph{executable}.

A \emph{rewrite theory} is a tuple $\cR = (\Sigma, E \uplus B, R)$. The static behavior of the system  is modeled by 
the order-sorted equational theory
 $(\Sigma, E \uplus B)$ and the dynamic behavior  by the set of rewrite rules $R$.  These components are explained below. \\
 
\noindent\textbf{Equational theory.} The signature $\Sigma$ defines a set of typed operators  used to build the terms of the language (i.e, the syntax of the modeled system).   $E$ is a set of  equations
over $T_\Sigma$ (the set of terms built from $\Sigma$)
of the form $t = t' ~ \mathbf{if} \phi$.
The equations specify the algebraic identities that terms of the language must satisfy (e.g., $|\epsilon| = 0$ and $|ax|  = 1+ |x|$ where $|x|$ denotes the length of $x$).  $B$ is a set of structural axioms
over $T_\Sigma$ for which there is a
finitary matching algorithm. Such axioms include   associativity, commutativity, and
identity, or combinations of them. 
For instance,  $\epsilon$ is the identity for concatenation and then, modulo this axiom, the term $x\epsilon$ is equivalent to $x$. 
The equational theory associated to $\cR$ thus defines deterministic and finite computations as in a functional programming language.

Our semantics for PEGs  starts by defining an appropriate equational theory
 or functional module --\lstinline{fmod} -- in Maude's terminology. Here we define the set of terminal and non-terminal symbols and strings: 

\begin{lstlisting}

fmod PEG-SYNTAX is
 sort NTSymbol . --- Non-terminal symbols
 sort TSymbol .  --- Terminal symbols 
 sort Str .      --- Strings 
 sorts TChar TExp . --- Chars and character classes
  ... --- to be completed/explained below
endfm  
\end{lstlisting}

First,  sorts (or types) are declared. As we shall see,   \code{TChar} is the basic block for building terms of type  \code{Str}  (by juxtaposition/concatenation)  and \code{TExp} will be populated with the usual patterns such as  \code{[0-9]}, \code{[a-z]}, etc.

As we said before, the equational theory is ordered-sorted. This means that besides having many sorts, 
there is a partial order on  sorts defining a sub-typing relation:
\begin{lstlisting}

 subsorts String < TChar  < Str .   --- Sub-typing
 subsort TChar TExp < TSymbol .
 subsort Qid < NTSymbol .
\end{lstlisting}

The sort \code{String} is part of the standard library of Maude and it represents the usual strings in programming languages. Hence, ``('' is a Maude's \code{String} and also a \code{TChar} in this specification (being \code{String} the \emph{least sort}). 
Note  that terms of type  \code{TChar} and \code{TExp} are also terminal symbols due to the subsort relation. 
Hence, examples of valid terminal symbols  are \lstinline{";", "c", "3", [0-9]}, etc. 
A Qid is a \emph{qualified identifier} of the form \lstinline{'X}  (using an apostrophe at the beginning of the expression). Examples of non-terminal symbols are  \lstinline{'Begin},  \lstinline{'Statement}, etc. In order to make the presentation cleaner, in the forthcoming examples, we shall 
omit the apostrophe and write simply   \code{Begin} instead of \code{'Begin}.


Besides the sorts, the signature also specifies the functional symbols or operators  that  define the syntax of the model. Here some examples of constructors for the type \code{TExp}:\\
\begin{lstlisting}
 op [.] : -> TExp . --- Any character
 ops [0-9] [a-z] [A-Z] ... : -> TExp .
\end{lstlisting}

All these operators are constants (functions without parameters) of type \code{TExp}. 
For the sort  \code{Str}, we have: 
\begin{lstlisting}

 op eps : -> Str . --- empty string
  --- Strings (concatenated with whitespace) 
 op __ : TChar Str -> Str [right id: eps ] . 
\end{lstlisting}

 \lstinline{eps} is the defined constant to denote the empty string $\epsilon$. 
In Maude, ``\lstinline{_}'' denotes the position of an argument. The operator \lstinline{op __} (usually called empty syntax in Maude)  
receives two parameters and returns a  \code{Str}. Note 
that \code{eps} is the right identity (an \emph{axiom} associated to \code{op __}) for this operator (i.e., $x\epsilon \equiv x$). 
As an example, the term \code{"(" "3" "+" "2" ")"} is an inhabitant of 
\code{Str}. For the sake of readability, when no confusion arises, we shall omit  quotes on characters and simply write   \code{( 3 + 2 )}. 

Deterministic and terminating computations are specified via equations. In fact,  an equational theory is  executable 
only if it is terminating, confluent  and sort-decreasing \cite{DBLP:conf/maude/2007}. Under these conditions, 
the mathematical meaning of the equality $t\equiv t'$ coincides 
with the following strategy:  reduce $t$ and $t'$ to their 
unique (due to termination and confluence) normal forms $t_c$ and $t_c'$ 
using the equations of the theory as \emph{simplification rules} from left to right. Then, $t \equiv t'$ 
iff $t_c =_B t_c'$ (note that $=_B$, equality modulo $B$,  is decidable 
since it is assumed a finitary matching algorithm
for   $B$). 

As an example, 
the following specification checks whether a terminal symbol matches a \code{TChar}:

\begin{lstlisting}

  op match : TSymbol TChar -> Bool .
  vars tc tc' : TChar .  --- logical variables
  eq match(tc, tc')     =   tc == tc' . 
  eq match([.], tc)     =   true .
  eq match([0-9], tc)   =   tc >= "0" and-then tc <= "9" .
  ...
\end{lstlisting}

The \code{==} symbol is built-in Maude   (in this case, checking whether two   \code{String}s are equal). 
Variables in equations are implicitly  universally quantified and  the second equation  must be read as 
\[
(\forall \texttt{tc}:\texttt{TChar}). (\texttt{match}([.], \texttt{tc}) = \texttt{true})
\]
 An equation then rewrites its left hand side into the right hand side \emph{simplifying}
 terms. For instance,  
 \[  \mbox{the term 
 \code{match([.], "a" )}  reduces into   \code{true} }
\]

The syntax of parsing expressions is defined as follows:
\begin{lstlisting}

sort Exp . --- Parsing expressions 
op emp : -> Exp . --- the empty expression
--- Terminal and non-terminal symbols are expressions
subsort TSymbol NTSymbol < Exp .
op _._ : Exp Exp -> Exp . --- Sequence
op _/_ : Exp Exp -> Exp . --- Ordered choice
op _* : Exp -> Exp .      --- Repetition
op !_ : Exp -> Exp .      --- Negative predicate 
\end{lstlisting}

Note that sequential composition is represented as \code{e.e'}. For the sake of readability, we shall omit the ``.''  and simply write  \code{e e'} when no confusion arises. 

It is also possible to specify derived constructors by defining the appropriate equations.  For instance, the \emph{and predicate} $\& e$ 
  can be defined as  $!! e$ \cite{DBLP:conf/popl/Ford04}: 
\begin{lstlisting}

op &_ : Exp -> Exp .    --- derived operator  
eq & e = ! ! e .        --- equation given meaning to & e
\end{lstlisting}

Hence, $\& e$ attempts to match $e$ and, if it succeeds,  then it backtracks to the starting point (not consuming any part of the input). 
More examples of derived constructors are 
\begin{lstlisting}

op _? : Exp -> Exp .  --- zero or one
op _+ : Exp -> Exp .  --- one or more
eq e ? = e / emp .
eq e + = e . e * .
\end{lstlisting}

Production rules ($A \leftarrow p$), grammars (as sets of production rules) and parsing states are defined as follows: \\
\begin{lstlisting}
sorts  Rule  Grammar State .
subsort Rule < Grammar . --- Every rule is a grammar
op _<-_ : NTSymbol Exp -> Rule .
op nil : -> Grammar . --- empty set of rules
--- Concatenating rules
op _,_ : Grammar Grammar->Grammar [comm assoc id:nil] .
--- Parsing states
op _[_]_ : Grammar Exp Str -> State . --- G[e]x
op fail : -> State .  --- fail 
subsort Str < State . --- x is also a state
\end{lstlisting}

Note  the axioms imposed on the operator for ``concatenating'' rules: the order of the rules  (\lstinline{comm} for commutativity) as well as parentheses (\lstinline{assoc} for associativity)
are irrelevant.  Moreover, \lstinline{nil} can be removed/added at will. 

Thanks to the almost zero representational distance \cite{meseguer-twenty-2012} 
between the   above specification
and the syntax of PEGs, we can see the system and its specification as isomorphic structures (with slightly different notations).
From now on, it should be clear from the context that   $G$ in  the expression ``$\PEGstateD$''
is a grammar while \code{G} in ``\code{G[ e ] x}'' 
is a term of sort \code{Grammar}. As noticed, we have consistently used the \texttt{monospace} font for objects in the specification. 

This concludes the specification of the equational theory defining the syntax and basic operations on PEGs. The next step is to define  rules encoding the semantics for PEG's constructors.\\

\subsection{A First Rewrite Theory}\label{sec:firstrewrite}

\noindent\textbf{Rewriting rules.}
 The last component $R$ in the rewrite theory 
 $\cR = (\Sigma, E \uplus B, R)$   is a finite set of rewriting rules of the form $t \to t'~\mathbf{if} \phi$.  A rule defines a state transformation and $R$ models  the dynamic behavior of the system (which is not necessarily deterministic, nor terminating). 
In Maude, rewrite theories are defined in modules: 
\begin{lstlisting}

mod PEG-SEMANTICS is
 pr PEG-SYNTAX . --- Importing PEG-SYNTAX
 var x : Str .    
 var G : Grammar .
 rl [empty] : G[ emp ] x => x .
 ...
endm
\end{lstlisting}

   ``\code{empty}'' is the name of the rule and variables    are implicitly  quantified. Hence, the above rule must be interpreted as:
   
\[
(\forall \texttt{G}:\texttt{Grammar}, \texttt{x}:\texttt{Str}) . (\texttt{G} [\texttt{emp}]~ \texttt{x}  \Rightarrow  \texttt{x}  )
\]

This rule reflects the behavior of the rule {\bf empty} in Fig. \ref{fig:sempeglab}. The 
rules for terminal  symbols are:
\begin{lstlisting}

rl [Terminal12] :  G[t] tc x => 
   if match(t,tc)  then x else fail fi .
rl [Terminal3] :  G[t] eps  => fail .
\end{lstlisting}
The first rule encodes both {\bf term.1} and  {\bf term.2}
where  the outcome depends on whether \code{t} matches \code{tc}. The second rule encodes the behavior of {\bf term.3}. 

The other rules are conditional rules (\lstinline{crl} instead of \lstinline{rl}) where the transition happens only if the condition after the symbol \lstinline{if} holds:

\begin{lstlisting}

crl [NTerminal] : ( G, N <- e ) [N] x => S
     if  (G, N <- e) [e] x => S .
crl [Seq1] :  G  [e . e'] x => S 
     if  G[ e ] x  =>  y     /\     
         G[ e' ] y  => S .
crl [Seq2] : G[e . e'] x => fail    if G[ e ] x  => fail .

crl [Choice1] :  G[ e / e'] x => y  if G[ e ] x  =>  y .
crl [Choice2] :  G[e / e'] x => S 
     if G[ e ] x  => fail /\ 
        G[ e' ] x  =>  S .

crl [Star1] : G[ e * ] x => x if G[ e ] x  => fail .
crl [Star2] :  G[ e * ] x  => z  
     if G[ e ] x  => y /\ 
        G[e *] y  => z .

crl [Neg1] : G  [ ! e ] x => x if G[ e ] x  => fail .
crl [Neg2] : G  [ ! e ] x => fail if G[ e ] x  => y .
\end{lstlisting}

In the case of \lstinline{[NTerminal]}, due to the axioms of the operator \lstinline{_,_} (i.e., \lstinline{[comm assoc id:nil]}), the order of the rules is irrelevant. Hence, if the current expression is \lstinline{N} (\code{N} is a variable of type \code{NTerminal}), this rule will unfold the production rule $N \leftarrow e $
 and  try to match  \lstinline{e} with the current input.

In   \lstinline{Seq1}, it must be the case that \lstinline| G[e] x| reduces to  \lstinline|y| and the configuration 
\lstinline|G[e'] y| must reduce to the state \lstinline{S}. This reflects exactly the behavior of the rule {\bf seq.1}. The other rules can be explained similarly.

A rewrite theory $\cR$ induces a rewrite relation $\rwR$ on
$T_{\Sigma}(\mathcal{X})$ (the set of terms build from $\Sigma$ and the countably set of variables $\mathcal{X}$) 
defined for every $t,u \in T_\Sigma(\mathcal{X})$ by $t
\rwR u$ if and only if there is a rule $(\ccrl{l}{r}{\phi}) \in
R$ and a substitution $\func{\theta}{\mathcal{X}}{T_\Sigma(\mathcal{X})}$ satisfying $t
=_{E \uplus B} l\theta$, $u =_{E \uplus B} r\theta$, and $\phi\theta$
is (equationally) provable from $E\uplus
B$~\cite{bruni-semantics-2006}.
In words, $t$ matches module $E\uplus B$ the left hand side of the rewrite rule under a suitable substitution and then evolves into the right hand side of the rule if the condition $\phi\theta$ holds. 
The relation
$\rwRR$ is  the reflexive and transitive closure of $\rwR $. Moreover, we shall use $t \rwRB t'$ to denote that $t \rwRR t'$ and   $t'\not\rwR$, i.e., $t$ reduces in zero or more steps to  $t'$ 
 and it  cannot be further reduced ($t'$ is called a \emph{normal  form}).

We shall use $\rwpeg$ to denote 
the induced rewrite relation from the above described theory. As a simple example, 
note that \lstinline|G["a"."b"] "a" "b"| $\rwpegB$ \lstinline{eps}. 
Note also  that the states \lstinline|x| (for any \code{Str} \code{x}) and \code{fail} are the only normal forms for $\rwpeg$.

\begin{theorem}[Adequacy]\label{th:adeq}
Let   $G$ be a grammar, $e$ a parsing expression and  $x$ a string. Then the following holds: 

\noindent\textbf{(1)} $\PEGstateD  \Lp y $ \ \ \quad \ \ \ \ iff 
\qquad \code{G[e]x} $\rwpegB$ \code{y}.
 
\noindent\textbf{(2)}
 $\PEGstateD  \Lp \Lfail  $ \quad iff
\qquad \code{G[e]x} $\rwpegB$ \code{fail}.
\end{theorem}
\begin{proof}
(sketch). We can show that, for any state $S$, 
$\PEGstateD  \Lp S $ iff  \code{G[e]x}$\rwpegB$ \code{S}. This 
discharges both (1) and (2). Note that such $S$ must be a normal form and then, either $S=\Lfail$ or $S=x$ for some string $x$. 
The implication ($\Rightarrow$) is proved by induction on the height of the derivation of  $\PEGstateD  \Lp S $.
For the ($\Leftarrow$) side, assume that $t = $\code{ G[e]x}$\rwpegB$ \code{S}.
This means that there exists $n>0$ (since $\PEGstateD$ is not a normal form)  and a derivation of the form  
$t = t_0 \rwpeg t_1 \rwpeg \cdots \rwpeg t_n\not\rwpeg$. 
Due to the side conditions in the rules, each step on this derivation may include the application of several rules. Hence, we proceed by induction on $m$ where $m$ is the total number of rules applied (including side conditions) in the above derivation. 
See  \S\ref{app:proofThAdeq1} for
more details. 
\end{proof}

We can use  this  theory as a (naive) PEG parser. It suffices to use the Maude's rewriting mechanism  on a PEG state: 
\begin{lstlisting}

rewrite ('A <- "a" , 'B <- "b") ['A .'B] "a" "b" "c" .
result State: "c"
\end{lstlisting}

As noticed, our  specification 
follows exactly the semantic rules in Fig. \ref{fig:sempeglab},  which makes it easier to prove its correctness. 
However, the resulting rewrite theory  is not completely satisfactory due to the  conditional rules 
used in choices, sequences and negative predicates: a whole derivation must be built to check whether the associated conditions are met or not. 
As we know, PEGs are deterministic and 
it should be possible to 
make a more guided decision on, e.g., 
whether 
\textbf{ord.1} or \textbf{ord.2}
should  be used on a given string.   For that, the strategy on $e_1 \slash e_2$ could be: reduce  $e_1$ and, in the end, decide whether or not $e_2$ is discarded. 
This reduction strategy  resembles  a 
small-step semantics and we shall explore it in the next section. We shall show that the new specification is more efficient for  checking  $x \in \Lang{G}$
and, more importantly, it is  appropriate for other (symbolic) analyses. 

\subsection{A More Efficient Rewrite Theory}\label{sec:moreefficient}
Semantic and logical frameworks are adequate tools for specifying and reasoning about different systems. By choosing the right abstractions in the framework, we can also have efficient procedures for such specifications. In the following, we introduce a rewrite theory  
that gives an alternative representation of   choices and failures in PEGs. 

We first introduce semantic-level constructs for negation, choices, sequential composition and replication: 
\begin{lstlisting}

op NEG :    State State -> State [ frozen (2)]  .
op CHOICE : State State -> State [ frozen (2)] .
op COMP :   State State -> State [ frozen (2)] .
op STAR :   State State -> State [ frozen (2)] .
\end{lstlisting}

Note that such operators are defined on states. 
The attribute \lstinline{frozen} will be explained in brief. 

The rules \lstinline{empty}, \lstinline{Terminal12} and 
\lstinline{Terminal3} are the same as in $\rwpeg$. Note that those rules
are not conditional.  

Let us introduce the new rules:
\begin{lstlisting}

rl [NTerm] : (G, N <- e) [N] x => (G, N <- e) [e] x .
rl [Sequence] :  G[e . e'] x  => COMP(G[e]x ,  G[e']x ) .
rl [Seq1] : COMP(y ,  G[e']x) =>  G[e'] y .
rl [Seq2] : COMP(fail , S) => fail .
\end{lstlisting}

In \code{NTerm},   \code{N} is ``simplified'' to the corresponding expression \code{e}. 
The rule \code{Sequence} replaces the  PEG constructor  \lstinline{e.e'}   with the corresponding  semantic-level constructor  \code{COMP}. The meaning of \code{COMP} is given by the two last rules:  
 if the first component succeeds  returning  \code{y}, then the second expression \code{e'} must be evaluated against the input \code{y}. Moreover, 
if the first component fails, then the whole expression fails. 

The rules for the other constructors follow similarly: 

\begin{lstlisting}

rl [Choice] : G[e / e']x => CHOICE(G[e] x , G[e'] x) .
rl [Choice1] : CHOICE(x , S) => x .
rl [Choice2] : CHOICE(fail , S) => S .
rl [Star] : G[e *]x =>  STAR(G[e] x , G[e *] x) .
rl [Star1] : STAR(fail,  G[e *] x ) => x .
rl [Star2] : STAR(y ,  G[e *]x) => STAR(G[e] y ,  G[e *]y) .
rl [Negative] : G[! e] x  => NEG(G[e]x ,  G[! e] x ) .
rl [Neg1] : NEG(fail,  G[! e] x) => x .
rl [Neg2] : NEG(y, S) => fail .
\end{lstlisting}

The first rule in each case reduces the current expression to the corresponding constructor; the second and third rules decide whether 
the expression fails or succeeds. For instance, 
\lstinline{Choice1} 
discards the second alternative if the first one succeeds and \lstinline{Choice2}  selects the second alternative if the first one fails.

Rewriting logic is an inherent concurrent system where rules can be applied at any position/subterm of a bigger expression. Hence, the use of the attribute \lstinline{frozen} is important to 
guarantee the adequacy of our specification.
This attribute indicates  that   the second argument (of sort \code{State}) of the operators \code{COMP, CHOICE, STAR} and \code{NEG} cannot be subject of reduction. In words, the state \code{S'} in \code{CHOICE(S,S')}
is not reduced until the first component is completely reduced. 

We shall use $\Rwpeg$ to denote the induced rewrite relation of the above specification. 

\begin{theorem}[Adequacy]\label{th:adq2}
Let   $G$ be a grammar, $e$ a parsing expression and  $x$ a string. Then the following holds: 

\noindent\textbf{(1)}  \code{G[e]x} $\rwpegB$ \code{y} \qquad  \ \quad iff \qquad 
 \code{G[e]x} $\RwpegB$ \code{y} 
 
\noindent\textbf{(2)} \code{G[e]x} $\rwpegB$ \code{fail} \qquad iff \qquad 
 \code{G[e]x} $\RwpegB$ \code{fail} 
\end{theorem}
\begin{proof}
We shall  prove (1) and (2) simultaneously, i.e., we shall show that 
$\PEGstateD\rwpegB S $ iff 
$\PEGstateD\RwpegB S $ for $S \in \{\Lfail, x\}$. 

\noindent $(\Rightarrow)$
We proceed by induction on the total number of rules (including side conditions) used in the derivation $\PEGstateD \rwpegB S$. The result is not difficult by noticing that 
 in $\rwpeg$, the condition of  the rules are satisfied by shorter derivations and then, by induction, one can show that the second and third  rules in the respective cases for $\Rwpeg$ mimic the same behavior.
 
 \noindent $(\Leftarrow)$
 Assume that $\PEGstateD \RwpegB S$. There are no side conditions but there are some extra intermediate steps due to the semantic-level constructors \code{CHOICE}, \code{COMP}, etc and the rules
 \code{Choice}, \code{Sequence}, etc. We note that 
 \code{COMP(S, S')} is not a normal form: if $S$ is a normal form, then \code{Seq1} or \code{Seq2} are used to continue reducing the term. Also, due to the \code{frozen} attribute, \code{S'} is never reduced in the scope of the \code{COMP} constructor. Similar observations apply for  \code{CHOICE}, \code{NEG} and \code{STAR}. We then proceed by induction on the total number of steps needed  to show $\PEGstateD \RwpegB S$. The base cases are easy. Consider the inductive case when the derivation starts with the rule \code{Negative}. Due to the above observations, we are in the following situation: 
 \[
 \PEGstate{G}{!e}x \Rwpeg \texttt{NEG}(\PEGstate{G}{e}{x}, \PEGstate{G}{!e}{x} ) \RwpegR 
 \texttt{NEG}(S_c, \PEGstate{G}{!e}{x} ) \Rwpeg S' 
 \]

 where $S_c$ is a normal form and the  rule applied on that state  can be either \code{Neg1} or \code{Neg2} depending on $S_c$. 
 This means that there is a shorter derivation for 
 $\PEGstate{G}{e}{x}\RwpegB S_c$ and the result follows by induction and applying {\bf not.1}
or  {\bf not.2} accordingly. 
\end{proof}

From the above result, determinism and Theorem \ref{th:adeq}, we know that 
$\Rwpeg$ is a function. Also, a more direct proof of that is possible: By simply inspecting the  rules and noticing that all the left-hand sides   are pairwise distinct  and only the left-most state can be reduced (due to the \code{frozen} attribute).

The analyses we are currently working on (see \S \ref{sec:conc} and \S \ref{sec:symbolic}), use  symbolic techniques in Maude, as \emph{narrowing}, that do not support conditional rewrite rules. This is one of the reasons to prefer $\Rwpeg$ over $\rwpeg$. Moreover,   avoiding conditional rules, $\Rwpeg$ outperforms $\rwpeg$ for membership checking. As a simple benchmark,
consider the following grammar  recognizing the language $a^nb^nc^n$:
\begin{lstlisting}

  S <- ( &(R1 "c") ) ("a" +) R2 (! [.]) , 
  R1 <- "a" (R1 ?) "b", 
  R2 <- "b" (R2 ?) "c"
\end{lstlisting}

In $\rwpeg$, the instance   $n=6$  is recognized in  30 seg, and the instance $n=7$ in 4.1 min.   Hence, this specification will be of little help for  practical purposes. $\Rwpeg$ recognizes the instance 
 $n=1000$ in less than one second. 
Although some improvements can be done such as using indices to avoid carrying in each state the fragment of the string being processed, 
 state of the art parsers for PEGs (e.g.,
 Mouse \cite{Redziejowski2015MOUSEFP}, 
  Rats! \cite{DBLP:conf/pldi/Grimm06} and 
Parboiled \cite{DBLP:journals/corr/abs-1907-03436})
   can certainly do much better than that. Our goal is not to build a parser but to provide a formal framework for PEGs and explore reasoning techniques for it. 
Hence, we prefer to keep the specification as simple and  general as possible
to widen the spectrum of
analyses that can be performed on it. 
The reader may have probably recognized the simplicity of the formal specification, driven directly from the syntax and semantic rules presented in \S \ref{sec:pegs}.

\section{Backtracks and Cuts  in PEGs}\label{sec:sym}

This section builds on the previous rewrite theory 
  to
study backtracks  in ordered choices. For that, we first introduce semantic rules for the \emph{cut} operator  in \cite{DBLP:conf/paste/MizushimaMY10}. By 
studying formally the meaning of cuts,  we   propose a deeper cut operator (\S \ref{sec:global})  that is able to save more  computations. 
In \S \ref{sec:bound} we show how local and global/deeper cuts can be uniformly  introduced in PEGs and
some extra optimizations  are presented in 
 \S \ref{sec:simpl}. 
 Finally, \S \ref{sec:bench} reports some experiments on grammars annotated with cuts. 

 
\subsection{Semantics for Cuts} 
\label{sec:back}

 The cut operator, inspired in 
 Prolog's   cut,
 was proposed in  \cite{DBLP:conf/paste/MizushimaMY10} with the aim of   better  controlling   backtracks in PEGs. The idea is 
 to annotate the grammar $G$ with cuts 
 leading to a modified grammar $G'$   in such a way that the  language recognized by 
 $G$ and $G'$ is the same but 
 $G'$ may avoid 
 some backtracks in choices  during parsing.  In fact, as shown in \cite{DBLP:conf/paste/MizushimaMY10}, this technique allows for  dynamically reclaiming the unnecessary space for memoization in those branches. 
 
 Cuts, as proposed in  \cite{DBLP:conf/paste/MizushimaMY10}, 
  cannot be introduced on arbitrary positions of a parsing expression. They only make  sense when the backtracking mechanism needs to be controlled. Hence, 
 there is a restricted syntax for them: 
\[
e ::= ... \mid e_1 \Pcut e_2 \slash e_3 ~\mid~ (e_1 \Pcut e_2) *
\]
 The intended meaning for the $\Pcut$ operator is: 

 \noindent \textbf{-} In $e_1 \Pcut e_2 \slash e_3$, if $e_1$ succeeds, $e_2$ is evaluated and $e_3$ is never considered even if $e_2$ fails. 

 \noindent \textbf{-} In $(e_1 \Pcut e_2) *$, when   $e_1$ fails, the entire expression succeeds. If $e_2$ fails (after the successful matching of $e_1$) the whole expression fails.

Consider for instance the following grammar:
\begin{lstlisting}[escapeinside={/*}{*/}]

       S <- TIF /*$\Pcut$*/ "(" ...  /  TWHILE /*$\Pcut$*/ "(" ... / TFOR ...
\end{lstlisting}

Clearly, if the token ``\texttt{if}'' was read from the input, a failure occurring right after that point does not need to backtrack to consider the other alternatives, 
that inevitable will also fail.  The introduction of the cut guarantees that: (1) information about the other alternatives can be discarded once the rule \code{TIF} succeeds (thus saving space); and (2) there is no need to explore  other alternatives after failure (thus failing faster).  

Adding cuts to a grammar needs to be done carefully. For instance, the languages generated by  the expressions 
\[\qquad \qquad a\ e_1 \slash a \ e_2 \qquad\mbox{ and} \qquad a\ \Pcut\ e_1 \slash a \ e_2
\]
are not necessarily the same. Algorithms for 
adding cuts to grammars are  proposed in \cite{DBLP:conf/paste/MizushimaMY10}. 
Roughly,  the \texttt{FIRST} set is used to check whether the alternatives  are disjoint. 
If this is the case,  it is safe to introduce a cut. 

We can formally state the semantics of cuts in the scope of choices through the following rule: 
\begin{lstlisting}[escapeinside={/*}{*/}]

rl [Choice^] CHOICE(G[/*$\Pcut$ */ e] x , S) =>  G[e] x.
\end{lstlisting}

This rule reflects the fact that,  once the symbol $\Pcut$
is found in the context of an ordered choice, the second 
alternative (the variable \code{S} of type \code{State}) can be discarded. 

Let us analyze the case of  repetition. Since this operator can be encoded using recursion, one may simply use the previous definition. However, in our framework,    \code{*} is not a derived constructor and 
we are forced to give meaning to  $(e  \Pcut e') *$.
The first thing we note is that the expression $ (e e') *$ never fails but $(e \Pcut e') *$ may fail. For instance, 
$(ab)*$ recognizes the string ``abac'' (returning ``ac'')
while    $(a\Pcut b)*$ fails to match the same string. This means that  the use of cuts in the context of repetitions requires extra care. In fact, 
the grammar transformations proposed in 
\cite{DBLP:conf/paste/MizushimaMY10} 
translates expressions of the form $e* ~e'$ into $(!e'~\Pcut e)* e'$  whenever   $e,e'$ are both non-nullable expressions \cite{DBLP:journals/fuin/Redziejowski09} (i.e., they do not recognize the empty string) and $e,e'$ are disjoint. Hence, cuts in repetitions appear in  very controlled ways.

We can give meaning to $(e \Pcut e') *$ by dividing its execution into two steps: when processing $e$, 
failures cause the success of the whole expression; when processing $e'$, a failure must cause the failure of the whole expression. 
For that, besides the semantic-level operator 
\code{STAR} introduced in the previous section, we also consider the following one:

\begin{lstlisting}

op STAR^ :  State State -> State [ctor frozen (2)] .
\end{lstlisting}

This operator will be used to keep track of the execution of the  expression $e'$ as follows: 
\begin{lstlisting}[escapeinside={/*}{*/}]

rl [Star] :  STAR(G[/*$\Pcut$*/ e'] x' , S ) => STAR^( G[e'] x', S) .
rl [Star^] : STAR^(y , G[ e *] x) => STAR(G[e] y, G[e *] y) .
rl [Star^] : STAR^(fail, S) => fail .
\end{lstlisting}

The first rule makes the transition from  \lstinline|STAR| to \lstinline|STAR^|. This happens when the  current expression  is a  cut. The second rule models the recursive behavior: if $e'$ succeeds, then we are back into the \lstinline|STAR| state. The third rule reduces to $\Lfail$ if $e'$ fails. 

These rules formalize the behavior of cuts in agreement with the intended meaning  in  
\cite{DBLP:conf/paste/MizushimaMY10}. However, the treatment of cuts does not look uniform: the meaning of $\Pcut$ is given in the context of other parsing expressions and not in a general way. 
The next section shows that it is possible to give a 
more precise 
meaning to  failures, cuts and backtracks. On doing that, we propose a deeper cut operator 
with a more elegant semantics that, in some cases,  may avoid more backtracks than $\Pcut$. 

\subsection{Global Errors and Cuts}\label{sec:global}
The cut operator $\Pcut$ acts locally. Take for instance the  PEG 
\begin{lstlisting}[escapeinside={/*}{*/}]

  A <- T1 /*$\Pcut$*/ e1   /  T2 /*$\Pcut$*/ e2
  B <- T3 /*$\Pcut$*/ e3   /  T4 /*$\Pcut$*/ e4
  C <- A / B
\end{lstlisting}

where \code{T1,T2,T3,T4} are lexical non-terminal symbols  and the 
grammar satisfies the unique token prefix condition (Def. \ref{def:unique}). Assume also that: the input starts with the string recognized by \code{T1};  the initial expression is \code{C}; and
the expression \code{e1}  fails to match the input. Due to the semantics of $\Pcut$, 
the second alternative in \code{A} is discarded and   \code{A} fails.
However, such failure does not propagate
to \code{C} and the alternative \code{B}
is tried against the input.  Since all the tokens  are disjoint, after an attempt on  \code{T3} and \code{T4},  the expression \code{C}
finally fails.   The key point is: failures in sub-expressions are \emph{confined} and they do not \emph{propagate} to outer levels. 
In the following, we propose a new operator that acts as a cut but \emph{globally}. 

We first extend the set of states (Def. \ref{def:states})  with the constant \code{error}, denoting the fact that an \emph{unrecoverable} error has been \emph{thrown} and  therefore  the global expression must fail:

\begin{lstlisting}

op error : -> State .
\end{lstlisting}

Moreover, we add to the syntax of parsing expressions the operator $\throw$ that, intuitively, \emph{throws} an error: 
\begin{lstlisting}[escapeinside={/*}{*/}]

op throw : -> Exp .     --- Representation of /*$\throw$*/
op check : Exp -> Exp . --- Derived constructor
eq check(E) = E / throw . 
\end{lstlisting}

The expression 
\code{check(e)} tries to match the expression \code{e} on the current input. If it fails, an error is thrown.

As we saw, the definition of the cut operator $\Pcut$ is problematic since its behavior depends on the context where it is evaluated. The operator  $\throw$ can appear in any context and its semantics will be given in a uniform way. 
In fact, the definition is quite simple:  on any input, 
$\throw$   raises an error:

\begin{lstlisting}

rl [throw] : G[throw] x => error .
\end{lstlisting}

Since we have two failing states, namely  \code{fail} and \code{error}, the set of rules in   \S \ref{sec:moreefficient} must be extended 
accordingly: 

\begin{lstlisting}

rl [SeqE] :     COMP(error, S)    =>  error .
rl [ChoiceE] :  CHOICE(error, S)  =>  error .
rl [StarE] :    STAR(error, S)    =>  error .
rl [NegE] :     NEG(error, G[! e]x)  => x .
\end{lstlisting}

\begin{figure}

\resizebox{.48\textwidth}{!}{$
\begin{array}{c}
\infer[\mylabel{throw}]{G[\throw]  \; x \LpeP \Lerror}{}\\\\
\infer[\mylabel{seq.3}]{G[e_1 \; e_2]  \; x \LpeP \Lerror}{G[e_1]  \; x \LpeP \Lerror}\qquad
\infer[\mylabel{ord.3}]{\Matgk{e_1 \;\slash\; e_2}{x}{} \LpeP \Lerror}{\Matgk{e_1}{x}{} \LpeP \Lerror }\\\\
\infer[\mylabel{not.3}]{G[!e]  \; x \LpeP x}{G[e] \;  x \LpeP \Lerror}\qquad
\infer[\mylabel{rep.3}]{G[e*]  \; x \LpeP \Lerror}{G[e]  \; x \LpeP \Lerror}
\end{array}
$}
\caption{Global errors. $\LpeP$ extends $\Lp$ in  Fig \ref{fig:sempeglab}. 
 \label{fig:error}}
\end{figure}

These definitions allow for the propagation of the error to the outermost context  as the Example \ref{ex:error} below shows. 
From now on, we shall use $\Lpe$ to denote the induced rewriting relation from the  theory above. 
The new rules in natural-semantic style are in Figure \ref{fig:error}.

\begin{example}\label{ex:error}
Consider the following grammar   for \code{labeled}, \code{goto} and \code{break} statements: 
\begin{lstlisting}[escapeinside={/*}{*/}]

  St <- labeledSt / jumpSt 
  labeledSt <-  ID check(COLON statement) / 
	 	CASE constantExp COLON statement 
  jumpSt <- GOTO check(ID SEMICOLON) /  BREAK SEMICOLON
\end{lstlisting}
Let $G'$ be as $G$ but  removing the \code{check()} annotations. Since the token expressions 
\code{ID}, \code{CASE}, \code{GOTO} and \code{BREAK} are all disjoint, we can show that:

\noindent\textbf{(1)} \lstinline|G[St]x| $\LpeR$ \lstinline|error|
implies  \lstinline|G'[St]x| $\RwpegR$ \lstinline|fail|. 

\noindent\textbf{(2) } \lstinline|G'[St]x| $\RwpegR$ \lstinline|fail|
implies \lstinline|G[St]x| $\Lpe^* l\in\{$ \lstinline|fail, error|$\}$.

In (1), the $\Lpe$-derivation does not 
need to match the (useless) alternatives, thus 
failing in fewer (or equal) steps when compared to the corresponding   $\Lp$-derivation
\end{example}

In the  example above, the number of backtracks 
is reduced when processing syntactically invalid inputs. On valid inputs, the number of backtracks remains the same.  
When predicates are involved, it is also possible to save some few (unnecessary) backtracks. 

\begin{example}\label{ex:pallene}
According to the ISO 7185 and ISO 10206 standards, Pascal allows comments opened with \code{(*} and closed with \lstinline|}|.  Consider the following grammar:
\begin{lstlisting}

comment <- open  (!close [.]) * close
open <- "(" "*" / "{"
close <- "*" check(")") / "}"
\end{lstlisting}

On   input ``\{ comment *here* *)'', the first and second ``*'' fail immediately (without trying to match ``\}'') and \code{!close} succeeds faster. 
Of course, this does not save much effort. 
As another example, a typical rule for 
identifiers looks like this:
\begin{lstlisting}

ID <- ! KEYW [a-zA-Z] ([a-zA-Z0-9_] *)
\end{lstlisting}
Some cuts can be added to the definition of reserved words, thus making the above predicate to fail faster:

\begin{lstlisting}

KEYW <- 'a' 'n' check('d') / 'a' check('s') / 
	'b' 'o' check('o' 'l') / 'b' 'r' check('e' 'a' 'k')...
\end{lstlisting}
On input ``bot'',  \code{ID} succeeds and only the first three choices of \code{KEYW} will be evaluated. 
%
%
%

\end{example}

Similar to $\Pcut$, a careless use of  $\throw$ may change the language recognized. Indeed, the situation is aggravated by the fact that errors propagate to outer levels. 
For instance, if \code{St}  in Example \ref{ex:error} is extended with a third choice recognizing \code{ID},  a failure in \code{labeledSt} should not be propagated to \code{St} since the new third alternative 
 cannot be discarded. 
Hence, it is salutary  to control the propagation of failures: a failure in \code{labeledSt} should avoid the (unnecessary) matching of \code{jumpSt} but such failure must remain \emph{confined} to these two alternatives.  This local behavior (akin to the one of $\Pcut$)  in combination with global failures will be explored in the next section.

\subsection{From Global to Local Errors}\label{sec:bound}
Now we introduce a \emph{catch} mechanism
to control the  global behavior of $\throw$
and confine errors when needed. 
The following definitions extend the syntax of expressions and states: 

\begin{lstlisting}

op catch : Exp -> Exp .     --- Catch operator (syntax)
op CATCH : State -> State . --- Semantics (on states)
\end{lstlisting}

The intended meaning is the following. If the expression \code{e} succeeds, then \code{catch(e)} also succeeds. If \code{e} fails  with final outcome $l\in\{$\code{error, fail}$\}$, the expression \code{catch(e)} fails with output \code{fail}.  This is formalized with  the following rules:
\begin{lstlisting}

rl [Catch]  : G[catch(e)] x  => CATCH(G[e] x) .
rl [Catch1] : CATCH(x) => x .
rl [Catch2] : CATCH( fail )  => fail .
rl [Catch3] : CATCH( error ) => fail .
\end{lstlisting}

As in the previous cases, the first rule moves from the syntactic level to the semantic operator at the level of \code{State}s.  The last three rules act on normal forms
materializing the intuition given above: \code{error}s are transformed into \code{fail}ures. 

\begin{example}\label{ex:error2}
Consider the production rules for labeled and jump statements in Example \ref{ex:error} and the rules below:

\begin{lstlisting}

  S1 <- catch(labeledSt / jumpSt) / assignSt
  S2 <- catch(labeledSt) / jumpSt / assignSt
  assignSt <- ID EQ expr SEMICOLON  
\end{lstlisting}
Since the token \code{ID} appears in both \code{assignSt} and \code{labeledSt}, 
\code{error}s  must be confined to guarantee that 
the \code{check(.)} annotations do not modify 
the language recognized. 
\code{S1} and \code{S2} are two alternative solutions. 
Consider the input  string ``\code{x = 3 ;}''. 
In \code{S1}, \code{labeledSt}
produces \code{error} and \code{jumpSt} is not evaluated (global behavior). This error is confined 
in \code{S1} (local behavior) and 
\code{assignSt} successfully recognizes  the input. 
A similar situation happens in \code{S2} but, when \code{labelSt} fails, \code{jumpSt} is also evaluated (and fails). 
Consider now the input ``\code{goto l :}'' (note the ``:'' instead of ``;''). In \code{S1}, the failure of \code{jumpSt} is confined and \code{assignSt} is (unnecessarily) evaluated. 
In \code{S2}, such failure preempts the execution of 
\code{assignSt} and fails faster. 
\end{example}

Some other transformations on this grammar  are possible to control better the shared  \code{ID} token between \code{labeledSt}
and \code{assignSt}. 
For instance, it is possible to join together 
assignments  and label statements, with the 
clear disadvantage of making the grammar more difficult to read and understand. 
The use of \emph{catch} expressions is common  in programming languages, thus making the grammar above more comprehensible.  

We shall use $\Lpec$ to denote the resulting relation
extending $\Lpe$ with the rules above. 
Figure \ref{fig:pegcatch} depicts the corresponding rules in natural-semantic style. 

\begin{figure}
\resizebox{.46\textwidth}{!}{
$
\begin{array}{c}
\infer[\mylabel{catch.1}]{G[ \catch{e}]  \; x \LpecP y}{G[e]  \; x \LpecP y}\qquad
\infer[\mylabel{catch.2}]{G[ \catch{e}]  \; x \LpecP \Lfail}{G[e]  \; x \LpecP \Lfail}\\\\
\infer[\mylabel{catch.3}]{G[ \catch{e}]  \; x \LpecP \Lfail}{G[e]  \; x \LpecP \Lerror}
\end{array}
$
}
\caption{Semantic rules for the \emph{catch} mechanism. \label{fig:pegcatch}}
\end{figure}

\subsection{Simplifying States}\label{sec:simpl}
Unlike $\Pcut$,  the $\throw$ operator has its own meaning independent of the context where it appears. In particular, it does not   need to be used only 
inside a choice operator. For instance, 
 the three expressions
 $\throw$, $a \throw$ and $a \throw b$
 are all  valid expressions (whose language is empty). 
 It is also worth noticing   the difference on how the two cut operators deal with the remaining alternatives: when the expression
\code{check(e) / e'}
is evaluated,  the second alternative \code{e'} is  only discarded 
when \code{e}  actually fails. Differently, 
the expression $\Pcut e \slash e'$
 discards eagerly 
  the second alternative,  thus saving some space. One then may wonder whether it is possible to reconcile the idea of \emph{saving memory} that motivated the introduction of cuts
in   \cite{DBLP:conf/paste/MizushimaMY10} with the 
behavior proposed here for  $\throw$. 
It turns out that the rewriting logic framework can give us some ideas on how to do that as described below. 

We first introduce an alternative version of  \code{check}:

\begin{lstlisting}

op try : Exp -> Exp .     --- Syntactic level
op TRY : State -> State . --- Semantic level
\end{lstlisting}

Unlike \code{check(.)}, \code{try(.)} is not a derived operator and hence, its meaning needs to be specified:

\begin{lstlisting}

rl [Try] :  G[try(e)] x => TRY(G[e] x) .
rl [Try1] : TRY(x) => x .
rl [Try2] : TRY( fail )  => error .
rl [Try3] : TRY( error ) => error .
\end{lstlisting}

Note the duality between \code{TRY} and \code{CATCH}: the former converts  \code{fail}ures  into \code{error}s
and the latter maps \code{error}s into \code{fail}ures. 
 Intuitively, \code{try(e)} succeeds if \code{e} succeeds and fails with \code{error} if \code{e} fails.
 We shall use $\Lptc$ to denote the extension of $\Lpec$ with the  rules for \code{TRY}. The additional rules are also depicted  in  Fig. \ref{fig:pegtry}. 

As a simple example,  the languages recognized by the expressions 
\code{"a"} and  \code{try("a")}
are the same. However, the final failing states are different on input "b":
\code{try("a")} ends with an \code{error} (that nobody caught).

\begin{figure}
\resizebox{.42\textwidth}{!}{
$
\begin{array}{c}
\infer[\mylabel{try.1}]{G[ \try{e}]  \; x \LptcP y}{G[e]  \; x \LptcP y}\qquad
\infer[\mylabel{try.2}]{G[ \try{e}]  \; x \LptcP \Lerror}{G[e]  \; x \LptcP \Lfail}\\\\
\infer[\mylabel{try.3}]{G[ \try{e}]  \; x \LptcP \Lerror}{G[e]  \; x \LptcP \Lerror}
\end{array}
$
}
\caption{Semantic rules for the \emph{try} mechanism. \label{fig:pegtry}}
\end{figure}

The following equivalences  are an easy consequence from  the definition of $\Lptc$ (and determinism):  

\begin{enumerate}
 \item \code{catch(try(e))} $\equiv$ \code{e} \ \quad\qquad\qquad\qquad\ \ \ \ \ \  (cancellation) 
 \item \code{catch(e) catch(e')} $\equiv$ \code{catch(e.e')}. \ \ \  \ \ \ \quad(distributivity)
 \item \code{catch(! e)} $\equiv$ \code{!catch(e)} \qquad \qquad \   \ \ \ \ \  \ \quad(distributivity)
\end{enumerate}

Note that, in general, \code{catch(.)} does not distribute on choices.
As a simple counterexample,
$\catch{\throw \slash \epsilon}$  always reduces to $\Lfail$ while $\catch{\throw} \slash \catch{\epsilon}$  succeeds on any string. 
Note also that,  even if $\catch{\try{e}}$ and $e$ recognize the same language ($\catch{e}$ and $\try{e}$ accepts $x$ iff $e$ does),  their failing states are different.
Hence, $e \slash e'$ is not necessarily equivalent to the expression 
$\catch{\try{e}} \slash e'$ (i.e, $\equiv$ is not a congruence).

Now let us come back to the problem 
of saving the space needed to store backtracking information that will never be used. 
As explained in \S \ref{sec:rw}, equations in the equational theory can be used to ``simplify'' terms. 
The idea is to  extend  the equational theory, 
so   that we can add some structural rules that govern the state of the parser. The simplification proposed is the following: 

\begin{lstlisting}

eq [simpl] : CHOICE(TRY(S), S') = TRY(S). 
\end{lstlisting}

Here, the whole state \code{S'} can be ignored since it will never be evaluated. Note that this simplification is sound due to the semantics of \code{CHOICE}: if $S$ fails then, \code{TRY(S)} reduces to \code{error} and \code{S'} will  never be evaluated.
Note that such simplification applies only when \code{TRY} is in the immediate scope of a \code{CHOICE}. 
Moreover, \code{[simpl]}  may simplify several choices. For instance, in this theory,  \code{CHOICE(CHOICE(TRY(S), S'), S'')} 
is equal to \code{TRY(S)} (discarding \code{S'} and \code{S''}). 

Let $\Lptc'$ be the extension of
$\Lptc$ with the equation above. 

\begin{theorem}\label{th:adc-simpl}
\code{G[e] x}
 $\LptcB$ \code{S}
\qquad iff \qquad
 \code{G[e] x}
  $\LptcBP$
\code{S}. 
\end{theorem}
\begin{proof}
Observe that 
the term \code{TRY(S)} either reduces to $x$ or fail with label \code{error}. In the first case, \code{CHOICE(TRY(S), S')} reduces to $x$ and, in the second case, reduces to \code{error}. Hence, the result depends only on the outcome of \code{TRY(S)}. 
\end{proof}

\subsection{Experimental Results}\label{sec:bench}

In this section we present some benchmarks performed on grammars manually annotated with the 
cut constructors proposed here. When annotating a grammar, given a concatenation $e_1\,e_2$,
the general idea is to annotate the symbols in $e_2$ since those in $e_1$ match at least one
input symbol (if $e_1$ is not nullable). When $e_1$ matches the input, we are in a (local or global) right path, so we can safely discard other alternative paths and avoid backtracking when $e_2$ fails. 

The Maude's specification can be seen as an abstract machine whose derivation steps  correspond, approximately, to the operations performed by a parser. We shall report the number of entry rules  that need to be applied to reduce $G[e]x$ into a final state. 
Such number corresponds to the number of times a PEG constructor  is evaluated. 
For instance, if $G=\{A \leftarrow a\}$, the expression $G[Ab] ab$ reduces to $\epsilon$ in 4 steps: $1_\times$ for sequential composition, $2_\times$ 
for the terminal rule and  $1_\times$ 
for the non-terminal rule. 
With this methodology,  it is not surprising that, on valid inputs,  the annotated grammar   performs more steps since entering on $\try{e}$  also counts as a step. 
We also report  some time measures
to show that it is feasible to use our formal  framework for concrete experiments. \\

%
%
%
%
\begin{table*}[!t]
\caption{Experiments with JSON and C89. \label{table-results}}

\begin{tabular}{| r |  r |  r |  r |  r |  r |  r |}
\hline
   \multicolumn{7}{|c|}{\textbf{JSON}}\\\hline
\hline 
\textbf{File}& \multicolumn{3}{|c|}{\textbf{Non-valid inputs}}
& \multicolumn{3}{|c|}{\textbf{Valid inputs}}\\\hline
    \textbf{Bytes} &
\textbf{Grammar}
 & 
\textbf{Grammar+cuts}
 & 
\textbf{\% }
 & 
\textbf{Grammar}
 & 
\textbf{Grammar+cuts}
&
\textbf{\% }
 \\
\hline 

276 & 9,545 & 7,339 & 23.1\% & 2,096 & 2,103 & -0.3\%\\\hline
872 &12,682 & 10,141 & 20.0\% & 6,968 & 6,989 & -0.3\%\\\hline
1k & 34,149 & 29,092 & 14.8\% & 8,276 & 8,301 & -0.3\%\\\hline
2k & 27,075 & 23,180 & 14.4\% & 12,718 & 12,745 & -0.2\%\\\hline
4k &46,611 & 42,614 & 8.6\% & 22,993 & 23,042 & -0.2\%\\\hline
2k &112,805 & 108,687 & 3.7\% & 25,439 & 25,452 & -0.1\%\\\hline
6k &61,323 & 57,369 & 6.4\% & 39,700 & 39,802 & -0.3\%\\\hline
10k &72,036 & 68,265 & 5.2\% & 62,971 & 63,082 & -0.2\%\\\hline
14k &117,316 & 113,892 & 2.9\% & 90,022 & 90,174 & -0.2\%\\\hline\hline

\multicolumn{7}{|c|}{\textbf{C89 (Git files)}}\\\hline\hline
\textbf{File}& \multicolumn{3}{|c|}{\textbf{Non-valid inputs}}
& \multicolumn{3}{|c|}{\textbf{Valid inputs}}\\\hline
    \textbf{lines} &
\textbf{Grammar}
 & 
\textbf{Grammar+cuts}
 & 
\textbf{\% }
 & 
\textbf{Grammar}
 & 
\textbf{Grammar+cuts}
&
\textbf{\% }
 \\
\hline 

trace (329) &3,770,039 & 3,692,000 & 2.1\% & 1,222,637 & 1,224,472 & -0.2\%\\\hline
tree (131) &1,481,679 & 1,317,197 & 11.1\% & 240,599 & 240,995 & -0.2\%\\\hline
version (35)& 826,840 & 798,336 & 3.4\% & 165,720 & 165,978 & -0.2\%\\\hline
walker (144) &4,177,590 & 4,051,591 & 3.0\% & 787,573 & 789,148 & -0.2\%\\\hline

\end{tabular}

\end{table*}

\noindent\textbf{JSON}. 
The main non-terminal of the JSON grammar tries all the possible shapes for a value: 

\begin{lstlisting}

value <- str / num / obj / arr / "true" / ... 
\end{lstlisting}

Strings start with (simple) quotes, numbers with digits, objects with ``\lstinline|{|'', arrays with ``\code{[}'', etc. Those tokens are unique and, once consumed,   a failure in the remaining expression should make the whole expression, including \code{value}, to fail.
This allows us to introduce some  cuts, for instance, 
\begin{lstlisting}

obj <- '{' pair (',' pair)* try('}')  / '{' try('}')
\end{lstlisting}
Hence, an open curly bracket without the corresponding closing one  causes a global failure.  
It is also possible to replace the second occurrence of \code{pair} with \code{try(pair)}: once \code{','} is consumed, the parser cannot fail in recognizing a pair.

We took some files  from repositories benchmarking JSON parsers
and produced  random invalid files
by  word mutations \cite{DBLP:conf/sle/RaselimoTF19}. 
More precisely, we randomly deleted some  (maximum  10) symbols
\code{']'},\lstinline|'}'|, \code{':'} and \code{','} 
in the corresponding file. 
For each of the 9 (correct) files, we generated 10 invalid files.

The results are in Table~\ref{table-results}.
The columns are:   the size of the file; the sum
of the steps in the 10 cases for  the grammar  and the annotated grammar; the percentage of reduced steps ($1.0 - Gcut/G$);
the number of steps in the grammar and the annotated grammar 
when processing the original (valid) file; and the percentage of the steps increased ($1.0 - Gcut/G$). 

Due to the format of JSON files,  
the same production rule is applied several times. Hence, if an error 
 occurs towards the end of the file, the $\try{\cdot}$ constructor will be invoked several times and, only in the end, it will save some backtracks. 
On a MacBook Pro (4 cores 2,3 GHz, 8GB of RAM) 
running Maude 3.0, processing in batch the 
9 valid files and the 90 invalid files
with the two grammars 
 takes 4.76s (average of 10 runs). \\

\begin{table}
\caption{Experiments for Pallene \label{tab:pallene}}
\begin{tabular}{| r |  r |  r |  r |}
\hline
\textbf{\# Files}
  &
\textbf{Grammar}
 & 
\textbf{Gram+cuts}
 & 
\textbf{\% }
 \\
\hline 
74 non-valid files & 124,130 &	102,056	& 17.8\% \\\hline
87 valid files & 270,783 & 272,611 & -0.7\%
\\\hline
\end{tabular}
\end{table}

\noindent\textbf{Pallene.} Consider now the grammar for Pallene  \cite{DBLP:conf/sblp/GualandiI18}, a  statically-typed programming language derived from Lua. It is interesting to see the granularity we can achieve with the cut mechanism proposed here. For instance, in the rules
\begin{lstlisting}

import <- 'local' try(name) try('=') 'import' ...
foreign <-'local' try(name) try('=') 'foreign' ...
\end{lstlisting}
 \code{'import'} does not raise an error when failing (and other alternatives are still available)  while  
a failure in  \code{name}
is global. 
Compare this behavior with the annotation $e \Pcut e'$: after $\Pcut$, it is not possible to control which sub-expressions of $e'$ can be considered as normal failures (where alternatives cannot be discarded) or  unrecoverable errors.

We took 74 invalid small programs ($< 60$ bytes per file) 
  that were proposed by the designers of Pallene as benchmarks  
  to test some specific features of the language.  We also processed 87 correct files from the same repository.

  The results are in Table \ref{tab:pallene}. Maude took 2.67s  to process the 322 tests (valid and invalid files with the two grammars). \\

\noindent\textbf{C89.} 
Finally, we considered the grammar of C89.
Here some examples of annotations: 
\begin{lstlisting}

enumerator <- ID '=' try(const-exp)  /  ID
case <- "case" try(const-exp ':' stat) 
\end{lstlisting}
When enumerating the elements of a \code{enum} expression, if
after a name there is an an assignment operator ('='),
then an expression must come after it. The second rule fails when the token \code{case} is consumed and any of the remaining expressions fail.

We took fragments of the files \code{trace.c}, \code{tree-diff.c}, \code{version.c}, \code{walker.c}
in the  implementation of Git (\url{https://github.com/git/git}). 
For each case, we generated 10 different files by randomly deleting curly brackets, parentheses and semicolons. 
The results are in Table \ref{table-results}. In this case, 
Maude took 91.3s (avg. of 10 runs)
to process the whole collection of files.

\section{Related Work}
\label{sec:related}

Hutton introduced the use of a new parser
combinator to make a distinction between an error
and a failure during the parsing process~\cite{hutton_1992}.

The idea of defining multiple kinds of failures for PEGs
was proposed in~\cite{DBLP:conf/sblp/MaidlMI13}. In this work,   the authors
introduced labeled failures as a mechanism to improve syntactic
error reporting in PEGs
by associating a different
error message for each label. This first formalization
provided a kind of choice operator that could handle 
set of labeled failures. These labeled choices can
also be used as prediction mechanism, where a label
indicates which alternative of a choice should be
tried~\cite{DBLP:journals/scp/MaidlMMI16}, allowing PEGs to
simulate the LL(*) algorithm used previously by
ANTLR~\cite{DBLP:conf/pldi/ParrF11}. 

However, this first labeled failures formalization
makes harder to recover from
an error in its own context, since  when 
a label is handled 
in a choice, the information about the error context is lost.
Because of this, attempts to deal with error recovery in PEGs
favored a semantics where the $\throw$ operator itself handles
the error~\cite{DBLP:conf/sac/MedeirosM18,DBLP:journals/scp/MedeirosJM20}.
To recover from an error, the $\throw$ operator tries
to match a recovery expression (a regular parsing expression).

The present work uses the idea of labeled failures to make
a distinction between local errors and global errors,
where both kind of errors may avoid unnecessary backtracking
when compared to a regular failure. The \emph{catch} and 
\emph{try} mechanisms can be simulated by using
the error recovery semantics of the $\throw$ operator. In the case
of $\catch{\cdot}$, the recovery expression of $\throw$ should be
an expression that always fail (e.g., $!a a$), thus resulting in
a regular failure. In case of $\try{\cdot}$, we simply do not
provide a recovery expression for $\throw$. 
Although $\throw$ can simulate \code{catch(.)} and \code{try(.)},
the use of these specific constructs helped us to show how to control
local and global backtracks in PEGs, and to properly formalize
the cut mechanism proposed in~\cite{DBLP:conf/paste/MizushimaMY10}.

Rewriting logic has been extensively used for specifying and verifying different systems \cite{meseguer-twenty-2012}. In the context of programming languages, it is worth mentioning the K Framework \cite{DBLP:journals/entcs/RosuS14}. Symbolic techniques in rewriting are currently the focus of intensive research (see a survey in \cite{DBLP:journals/jlap/DuranEEMMRT20}). In fact, one of the inspirations of this work came from an example 
of CFGs reported in \cite{DBLP:journals/jlap/AlpuenteCEM20} (and also used in \cite{DBLP:journals/jlap/DuranEEMMRT20}). Roughly, 
\emph{narrowing} (rewriting with logical variables as in logic programming) was used  to 
explore symbolically 
the state of programs and apply  partial evaluation \cite{DBLP:conf/dagstuhl/1996pe} (a program transformation) to improve the efficiency of  Maude's specifications.


\section{Concluding Remarks}
\label{sec:conc}
We have proposed a  framework based on rewriting logic to formally study the behavior of PEGs. 
Relying on this machinery, we proposed a general view of cuts where local and global failures can be treated uniformly and with a clear semantics. Such operators have a pleasant similarity to the usual try/catch
expressions in programming languages, thus making it easier to understand and predict their behavior when designing a grammar. 
Our specification is not only a formal theory 
of the derivability relation in PEGs, but it is also an executable system. Based on it, we have tested some optimizations on grammars.


In \S \ref{sec:symbolic}, we report on a preliminary attempt of using the 
rewrite logic theory developed here, together with symbolic techniques, to propose a 
   derivative parser that can be used as basis for other analyses. In particular, 
 we show that it is possible to generate all the strings, up to a given length, from a grammar.   This should provide more tools for PEG's user to  verify whether a grammar 
correctly models a given language of interest. We still need to refine
this work and to further investigate how it compares to other works
that explored the use of derivatives in
PEGs~\cite{DBLP:conf/lata/Moss20,DBLP:journals/corr/abs-1801-10490}
and in CFGs~\cite{DBLP:conf/icfp/MightDS11}.

We also foresee  to use the symbolic traces to generate positive and negative cases for a grammar \cite{DBLP:conf/sle/RaselimoTF19}. It is also worth exploring if 
the symbolic semantics in 
\S \ref{sec:symbolic} can be useful for proving language equivalence on restricted fragments of PEGs (such problem is undecidable in general \cite{DBLP:conf/popl/Ford04}). 

The manual insertion of catch and try expressions   in a grammar may be
a tedious and error prone task. Moreover, it may hinder the grammar clarity. 
Fortunately, a good amount of theses annotations can be done automatically.
We are currently exploring different automatic labeling algorithms inspired by those
in \cite{DBLP:conf/paste/MizushimaMY10} and 
\cite{DBLP:journals/scp/MedeirosJM20}. Our framework can be handy in  proving  the correctness
of the resulting  algorithm (i.e., the language is preserved after the annotation). 
This approach based on the automatic annotations should
preserve the grammar clarity while still avoiding some amount of unnecessary
backtracking.

Finally, one of the main reasons for introducing  cuts in PEGs is to save memory in packrat parsers  \cite{DBLP:conf/paste/MizushimaMY10}. 
The cuts proposed here generalize this idea and avoid further backtracks. 
Since our specification does not build a memoization table, 
measuring the  usage memory
when running Maude is pointless. 
Hence, it may be worth implementing global cuts in a packrat parser in order to evaluate the impact of the optimization in terms of memory consumption.

\appendix
\section{Proofs of Adequacy Theorems}
\label{app:proofThAdeq1}
\noindent\textbf{ Theorem \ref{th:adeq} ($\rwpegB$ and $\Lp$ coincide)}
\begin{proof}
Recall that  strings (\code{x}) and \code{fail} are normal forms for 
$\rwpeg$. Also, if $\PEGstateD  \Lp S $, $S$ is either $\Lfail$ or some string $x$. We shall show that 
$\PEGstateD  \Lp S $ iff  \code{G[e]x}$\rwpegB$ \code{S} for any $S$. This 
discharges both (1) and (2).

\noindent(\textbf{$\Rightarrow$}). 
We proceed by induction on the height of the derivation of  $\PEGstateD  \Lp S $.
Assume a derivation of height 1.   Hence, either $S=x$ and 
 the rules {\bf empty} or {\bf term.1}
were used; or $S=\Lfail$ and {\bf term.2} or {\bf term.3}
were used. In the case {\bf empty},   $x=y$, $e=\epsilon$ and clearly 
\code{G[emp]x} $\rwpeg$ \code{x} by using the rule \code{empty}. In the case {\bf term.1}, $x=ay$  and $e=a$. The corresponding term \code{G[a] ay}  matches the left-hand side of  \code{Terminal12}, \code{match(a,a)} reduces to \texttt{true} and the whole term reduces to \code{y} as expected. The cases when $S = \Lfail$ are similar. 

Assume now that $\PEGstateD  \Lp S $ is proved with a derivation of height $>1$. He have 9 cases. 
Let us consider some of them since the others follow similarly. 
 If {\bf var} was used, $e=A$ and there is a (shorter) derivation of $\PEGstate{G}{P(A)}{x} \Lp S$. By induction, we know that \code{G [P(A)]x } $\rwpegB$ \code{S}. 
 Seen $G$ as a set of rules,  $G = G' \cup \{A\leftarrow P(A)\}$. If the current expression is \code{A} (a non-terminal), the only matching rule on the corresponding  term \code{t = (G', A <- P(A)) [A] x} is   \code{NTerminal} that unifies 
$G=G'$, $N=A$ and $e = P(A)$. Hence,
$t\rwpeg $\code{G[P(A)] x} which, in turns reduces 
to \code{S}. 
Assume now that the derivation ends with {\bf ord.2}. 
By induction we know that \code{G[e1]x}$\rwpegB$\code{fail} and also, \code{G[e2]x}$\rwpegB$\code{S}. 
There are two rules that match \code{e1 / e2} (\code{Choice1} and \code{Choice2}). However, 
the side condition in \code{Choice1} does not hold. Using \code{Choice2}, we know that 
\code{G[e1 / e2]x}$\rwpeg$\code{G[e2]x}
that, in turns, reduces to \code{S}. The other cases follow similarly.

\noindent(\textbf{$\Leftarrow$}). 
Assume that $t=$\code{G[e]x}$\rwpegB$ \code{S}.
This means that there exists $n>0$ (since $\PEGstateD$ is not a normal form)  and a derivation of the form  
$t = t_0 \rwpeg t_1 \rwpeg \cdots \rwpeg t_n\not\rwpeg$. 
Due to the side conditions in the rules, each step on this derivation may include the application of several rules. Hence, we proceed by induction on $m$ where $m$ is the total number of rules applied (including side conditions) in the above derivation. 
If $m=1$ then $n=1$ and either \code{empty}, \code{Terminal12} or \code{Terminal3} were used. The needed derivation in $\Lp$ results from applying the corresponding rule. 
If $m>1$ we have several cases depending on the rule applied on $t_0$.   Assume that the derivation starts with \code{Seq1}. This means that $e = e1 . e2$. Moreover, 
$t_1 = G[e2]y$ and, by the side condition of the rule, 
\code{G[e1] x => y}. By induction, $\PEGstate{G}{e_2}{y} \Lp S$. Since  \code{y} is a normal form,  \code{G[e1]x}$\rwpegB$ \code{y} with a smaller number of steps and, by induction,  $\PEGstate{G}{e_1}{x}  \Lp y $. Now we use {\bf seq.1} to conclude this case. 
The other cases follow similarly. 
\end{proof}

%

\section{Symbolic Executions and Derivatives}
\label{sec:symbolic}
One useful technique to predict the behavior of a grammar is to automatically generate strings from it and check whether the results match the intuitive behavior. There are  recently works   implementing derivative  parsers \cite{DBLP:journals/jacm/Brzozowski64} for PEGs to accomplish this task \cite{DBLP:journals/corr/abs-1801-10490,DBLP:conf/lata/Moss20}. One of the main difficulties  is precisely the backtrack mechanism in PEGs.  Hence, the algorithms need to keep track of the different branches \cite{DBLP:conf/lata/Moss20}
and compute possible over approximations \cite{DBLP:journals/corr/abs-1801-10490} of
different notions as testing whether an expression recognizes the empty string (which is undecidable in general \cite{DBLP:conf/popl/Ford04}). 
This section  shows that our formal specification can be  adapted to implement a derivative tool. By 
 using constraints, we give a symbolic and compact representation of the possible outputs (of a fixed length) 
 derivable from a grammar. \\
 
 \noindent\textbf{Constrained Strings}. 
 The first step is to replace the input strings with  sequences of \emph{constraints} $c_1 . c_2 ...$.
   Intuitively, $c_i$  represents the set of characters allowed in the $i$-th position of the string. 
   To formalize this idea, we rely on the concept of constraint systems, commonly used in constraint logic programming \cite{DBLP:conf/popl/SaraswatR90}.
   A constraint system provides a signature from which  constraints can be constructed as well as an entailment relation $\entails$ specifying interdependencies between these constraints. 
A  constraint represents a piece of \emph{partial information} and $c \entails d$ 
means that information $d$ can be deduced from $c$. In the following definition, $t$ is a terminal symbol (\code{TSymbol}),  
   $a$ is a character (\code{TChar}) and  $t_c$ represents a character class (\code{TExp}). Given a terminal symbol $t$, we shall use $dom(t)$ to denote the set of characters allowed by $t$ (e.g., $dom(``x")=\{x\}$, $dom(\mbox{[0-9]}) = \{0,1,...,9\}$, etc.).

 \begin{definition}[Constraint System]
Constraints, usually ranged over $c_1,c_2, etc.$ are built from:
\[
 \qquad \qquad c  ::= \true \mid \false \mid t \mid \negC t \mid c \wedge c
\]

The entailment relation  $\entails$
 is the least relation closed by the rules of intuitionistic logic and the following axioms: 
\[
\begin{array}{lll}
t \wedge t' \entails \false& \mbox{whenever } dom(t) \cap dom(t') = \emptyset \\ 
t \wedge \negC t' \entails \false & \mbox{whenever } dom(t) \cap \overline{dom(t')} = \emptyset \\ 
 t \vdash t'& \mbox{whenever } dom(t) \subseteq dom(t')
\end{array}
\]
\emph{Constrained string} are  sequences of constraints and $c^n$ denotes the strings containing $n$ copies of $c$. 
 \end{definition}

Let us give some intuitions. The entailment $3 \entails [0-9]$ holds since $3$ is stronger (i.e., it constraints more) than $\mbox{[0-9]}$. When more constraints are added, the set of characters allowed decreases. 
Hence, conjunction corresponds to intersection of  domains ($dom(c \wedge c') = dom(c) \cap dom(c')$). 
$\false$ is the strongest constraint (since $\false \entails c$ for any $c$) and it represents an  inconsistent state (with empty domain). $\true$ is the weakest (i.e., $c\entails \true$ for any $c$) constraint and it is equivalent to $[.]$. 
$\negC t$ can be interpreted as the complement of the domain of $t$. Finally, note that the constraint $``x" \wedge [0-9]$ (resp. $\mbox{``3''} \wedge \negC \mbox{[0-9]}$) is inconsistent in virtue of the first (resp. second) axiom of $\entails$. 

The above definition gives rise to the expected signature: 
\begin{lstlisting}

 --- Atomic constraints and constraints
 sort AConstraint Constraint .
 subsort TSymbol < AConstraint < Constraint .
 op ~_ : TSymbol -> AConstraint .
 ops ff tt : -> Constraint . --- True and False
 op _/\_ : Constraint Constraint  -> Constraint . 

\end{lstlisting}

Our derivative parser will refine, monotonically, 
cons-trained strings. This means  that, in each step, more information/constraints will be added.
In order to keep the output as short/readable as possible, we define some simplifications on constraints. For example,  if $d \entails c$, then $c\wedge d$ can be simplified to $d$ (since $d$ contains more   information than $c$): 
\begin{lstlisting}

eq c /\ ff = ff . --- ff entails any c
eq c /\ tt = c .  --- c entails tt
eq c /\ c = c .   --- idempotency
eq a /\ a' = if a == a' then a else ff fi . 
eq a /\ ~ a' = if a == a' then ff else a fi .
\end{lstlisting}

For instance: $``a" \wedge [a-z]$ reduces to $``a"$;  $``a" \wedge ``b"$ reduces to $\false$; $``a" \wedge \negC ``b"$ reduces to $a$; etc. 
More generally,  
\[
\begin{array}{lll}
c \wedge c' = c'  \mbox{\ \ \ whenever } dom(c') \subseteq dom(c)\\
c \wedge c' = \false  \mbox{\ \ \ whenever } dom(c) \cap  dom(c') = \emptyset
\end{array}
\]

Now we define constrained strings: 
\begin{lstlisting}

 sort CStr .  subsort  Constraint < CStr .
 op _._ : Constraint CStr -> CStr [ctor right id: nil] .
 eq ff . c . x = ff .  --- Simplifying inconsistencies 
 eq c . ff . x = ff .   
\end{lstlisting}

Note that, e.g.,  the constrained string $\mbox{[0-9]}.\false.``a"...$ collapses into $\false$. \\

\noindent\textbf{Derivative Rules.}
The derivative parsing is  obtained by adjusting the specification in \S \ref{sec:rw}. First, states take the form \lstinline|G[e]  x :: y| meaning that the constrained string \code{y} was already processed and \code{x} is still being processed. 
Moreover, the final states take the form \code{ok(x,y)} (\code{y} was successfully consumed and the non-consumed input is \code{x}) and \code{fail(x,y)} (\code{x} could not be processed and failed). 
The  main change occurs  in the terminal rules:
\begin{lstlisting}

rl [Terminal] : G[t] (c.x) :: y => ok(x , ins( c /\ t , y)) .
rl [Terminal] : G[t] (c.x) :: y => fail( (c /\ ~ t).x , y) .
\end{lstlisting}

Note that, unlike the theory in in \S \ref{sec:rw}, here we have some non-determinism. Under input $c.x$ (a constrained string) we have two possibilities: either the parser succeeds consuming $c$ and adding to it the fact  that the string must start with a character matching $t$; or, it fails and the string 
must start with a character in $\overline{dom(t)}$. 

The other rules must be adjusted accordingly
to carry the history of constraints accumulated so far.   For instance, the rule for failures in choices becomes:
\begin{lstlisting}

 rl [Choice] : CHOICE( fail(x, y) , G[e'] x':: y') => 
    G[e'] conj(x', concat(y, x)) :: y' .

\end{lstlisting}
This means that, if the choice failed, the second alternative $e'$ must be evaluated on a constrained string with the additional information accumulated during the failure of the first alternative $e$. This is the purpose of the function $conj$ that simply applies 
point-wise  conjunction on the elements of the strings. 

\begin{example}[Symbolic outputs]
For each expression,
we show the resulting final states of the form 
\code{ok(x,y)}. For readability, instead of \code{ok(x,y)} we write $y :: x$ ($y$ was consumed and $x$ was returned). Let's start with some simple cases:

\noindent\textbf{-} $[0-9]~  a$. Only one solution: \code{ [0-9].a :: tt.}. This means that any valid string must start with $x\in dom([0-9])$, continue with $a$ and then, any symbol is valid (for all $x$, $x\in dom(\true)$). 

\noindent\textbf{-} $!a\ b~ (c\slash d)$. 
Maude returns two solutions:
\begin{lstlisting}

(~ a /\ b).c :: tt 
(~ a /\ b). (~ c /\ d) ::tt
\end{lstlisting}
  which further simplifies to  \code{b.c::tt} and \code{b.d::tt}. 
  
\noindent\textbf{-} $!a\ a$: \code{no solution}

Consider now the rule: 
\begin{lstlisting}

NUMBER <-  [0-9]+ . ("." . ( ! "." . [0-9])+)? . 
\end{lstlisting}

For strings of at most 3 elements, we have 7 solutions:
\begin{lstlisting}

[0-9]::(~ "." /\ ~ [0-9]).tt --- ex. "3ax"
[0-9].[0-9]::(~ "." /\ ~ [0-9]) --- ex. "23x"
[0-9].[0-9].[0-9] --- ex "123"
[0-9]::("." /\ ~ [0-9]) ."." --- ex "1.."
[0-9]::("." /\ ~ [0-9]).(~ "." /\ ~ [0-9])--- ex."1.a"
[0-9].[0-9]::("." /\ ~ [0-9]) --- ex. "12."
[0-9].("." /\ ~ [0-9]) .([0-9] /\ ~ ".") --- ex. "1.2"
\end{lstlisting}
After ``::'' we have the part of the string not consumed. The first output reads: the first character must be a number (and it is comsumed); then, the second character cannot be a digit, nor ``.''; and the last element can be any symbol. Then, e.g., the string ``3ax'' is accepted retuning the suffix ``ax''. 

The PEG for $a^nb^nc^n$ in the end of \S \ref{sec:rw} returns, as unique solution, the expected strings. As already noticed in \cite{DBLP:journals/corr/abs-1801-10490}, the grammar for the same language proposed in \cite{DBLP:conf/popl/Ford04} is incorrect. For a length of 6, our tool finds the following  solutions:
\begin{lstlisting}

a a b b c c 
a a a a a a 
a a a a b c 
\end{lstlisting}
\end{example}
The above symbolic strings are generated by using the search facilities in Maude:
\begin{lstlisting}

search [n] G[e] tt^n =>* ok(x',y') such that c .
\end{lstlisting}

 meaning ``compute the first $n$ states of the form \code{ok(x',y')} that satisfy the condition \code{c}
and  can be obtained by constraining   the string $\true^n$''. 
Note that rewriting is not enough since the theory is not longer deterministic and different paths must be considered in the \code{Terminal} rules. 
The \lstinline|search| command  implements a breadth-first search procedure and, therefore, no solution is lost. 
Needless to say that the search space may grow  very quickly.
Hence, either the condition $c$ is used to filter some of the solutions (e.g., 
compute  the symbolic strings 
that start with a digit)
or specific/localized rules of the grammar are analyzed independently. 
%

We shall use $\SLp$ to denote the induced rewrite relation using the rules above. For a constrained string $s= c_1.c_2.\cdots$ and a string $x=a_1.a_2...$, we shall write $x \preceq s$ iff  for each $i$, $a_i \in dom(c_i)$ (i.e., $a_i$ is a legal character for $c_i$).

The next result shows that all the possible strings that can be generated from an expression $e$ are covered by the symbolic output and, moreover, all instances of a symbolic output are indeed valid outputs for $e$. 

\begin{theorem}[Adequacy]\label{th:sym-adq}
For all $G$, $e$, $x$ and $y$:

\noindent \textbf{Soundness:}  If $\PEGstate{G}{e}{xy} \Lp y$
 then, there exists $s_1$ and $s_2$ s.t. $\PEGstate{G}{e}{\true^{|xy|}} \SLpB \Lok(s_2,s_1)$ and $xy \preceq s_1::s_2$. 

\noindent \textbf{Completeness:} If $\PEGstate{G}{e}{\true}^n \SLpB \Lok(s_2, s_1)$ then, for all $xy \preceq s_1::s_2$, $\PEGstate{G}{e}{xy}\Lp y$. 
\end{theorem}

\begin{proof}
We shall prove the correspondence between
$\SLp$ and $\Rwpeg$. By Theorems \ref{th:adeq} and \ref{th:adq2}
, the result extends to $\Lp$. 

We shall consider an 
alternative version of $\Rwpeg$ that, on failures, returns   the non-consumed input. Hence, 
$\PEGstate{G}{e}{xy} \Rwpeg \Lfail(y)$ means that $x$ (a possible empty string) was consumed and $y$ could not be recognized. This will simplify the arguments below.
Our proof  considers the failing cases, i.e., we shall prove the following: 
\begin{itemize}
 \item Soundness:  If $\PEGstate{G}{e}{xy} \RwpegB y$
 then, there exists $s_1$ and $s_2$ s.t. $\PEGstate{G}{e}{\true^{|xy|}} \SLp \Lok(s_2,s_1)$ and $xy \preceq s_1::s_2$. Moreover, if $\PEGstate{G}{e}{xy} \RwpegB \Lfail(y)$
 then $\PEGstate{G}{e}{\true^{|xy|}} \SLp \Lfail(s_2, s_1)$ and 
 $xy \preceq s_1::s_2$. 
 \item Completeness: If $\PEGstate{G}{e}{\true}^n \SLp \Lok(s_2,s_1)$ then, for all $xy \preceq s_1::s_2$, $\PEGstate{G}{e}{xy}\Lp y$. Moreover, if
 $\PEGstate{G}{e}{\true}^n \SLp \Lfail(s_2,s_1)$, 
  for all $xy \preceq s_1::s_2$, $\PEGstate{G}{e}{xy}\Lp \Lfail(y)$.
 
\end{itemize}

The main difference between  $\SLp$ and $\Rwpeg$ is on the terminal rules. Namely,  $\SLp$ considers two cases: either the string contains the needed terminal symbol $t$ or it does not. In the second case, the derivation fails and adds the constraint $\negC t$. 

\noindent\textbf{Soundness}. 
We proceed by induction on the length of the derivation of $\PEGstate{G}{e}{x} \RwpegB S$. In the base case, either $e=\epsilon$ or $e=a$. In the case of a terminal symbol,
we have two possible outcomes: 
$\PEGstate{G}{a}{x} \Rwpeg x'$ (and $x=ax'$)
or $\PEGstate{G}{a}{x} \Rwpeg \Lfail(x)$ (and $x=bx'$ for $b\neq a$ or $x=\epsilon$). 
In the successful case, $n = |x|\geq 1$
and $\PEGstate{G}{a}{\true ^{|x|}}$ has two possible outcomes: $\Lok(\true^{n-1}, a)$ and $\Lfail((\negC a) \true^{n-1}, nil)$. Note that $ax' \preceq a\true^{|n-1|}$. The case when $x=bx'$ fails is considered in the second symbolic output: $bx' \preceq (\neg a)\true^{|n-1|}$ (note that, if $b\neq a$, then $b\in dom(\neg a)$). 
For the inductive case, we have several subcases. Consider the case when the derivation starts with \code{Choice}. Hence, $e=e_1\slash e_2$. There are two possible outcomes for $e_1$ and, by induction, both are instances of one of the symbolic outputs. The case when $e_1$ succeeds is immediate. 
Consider the case where $\PEGstate{G}{e_1}{xyz} \RwpegB \Lfail(yz)$ and $\PEGstate{G}{e_2}{xyz} \RwpegB z$. By induction, $\PEGstate{G}{e_1}{\true^n} \SLp \Lfail( s_ys_z ,s_x)$
where $s_y$ and $s_z$ explain the failure of 
$e_1$ and $s_x$ is in agreement with the part of the string consumed.   We know  that $x\preceq s_x$ and $yz \preceq s_ys_z$. By induction, we also  know that 
$\PEGstate{G}{e_2}{\true^n} \SLp \Lok(\true^{|z|}, w_xw_y)$. However, the semantics executes $e_2$
on $\true^n \wedge s_xs_ys_z$ (and not on $\true^n$). Since the semantics only adds new constraints to the sequence of constraints, we can show that 
$\PEGstate{G}{e_2}{\true^n  \wedge s_xs_ys_z} \SLp
 \Lok(s_z, (w_x \wedge s_x)(w_y\wedge s_y)$. Since
 $x\preceq s_x$ and $x\preceq w_x$, then 
 $x\preceq s_x \wedge w_x$. Similarly for $y$ and $z$ and the result follows. The other cases are similar. 
 
\noindent\textbf{Completeness}. 
We proceed by induction on the length of the derivation of 
$\PEGstate{G}{e}{\true}^n \SLpB S$.
For the base case, 
consider a one-step derivation and 
assume that $S=\Lok(s_y, c)$. Hence, $e=t$, $c=t$ and $s_y=\true^{n-1}$. If $xy\preceq t :: s_y$ then $x\in dom(t)$ and clearly $\PEGstate{G}{t}{xy}\Rwpeg y$. If $S=\Lfail((\negC t)s_y, nil)$ then
$e=t$ and $s_y=\true^{n-1}$. If $xy \preceq (\negC t)s_y$, $x\in dom(\negC t)$ and then, 
$x \notin dom(t)$. This means that $x$ does not match $t$ and $\PEGstate{G}{t}{xy}\Rwpeg \Lfail(xy)$. 

For  the inductive case, we have several subcases. Consider a derivation that starts with \code{Sequence}: 

$\PEGstate{G}{e_1.e_2}{\true^n} \SLpB \Lok(s_z, s_xs_y)$. 
By the definition of $\SLp$, we know that 
 $\PEGstate{G}{e_1}{\true^n} \SLpB \Lok(s_ys_z, s_x)$. Moreover, 
 
$\PEGstate{G}{e_2}{s_ys_z} \SLpB \Lok(s_z, s_y)$. By induction, we deduce

$\PEGstate{G}{e_1}{xyz} \RwpegB yz$. 
From a similar observation about conjunction of sequences (as done in the proof of Soundness), we can also show that 
$\PEGstate{G}{e_2}{yz} \RwpegB z$. 
We then conclude 

$\PEGstate{G}{e_1.e_2}{xyz} \RwpegB z$. 
   The other cases are similar. 
\end{proof}

\newpage


\bibliographystyle{ACM-Reference-Format}
\bibliography{ms}

\end{document}